\documentclass[letterpaper,10 pt,conference]{ieeeconf}  
\IEEEoverridecommandlockouts                              
\usepackage{cite}
\usepackage{graphics} 
\usepackage{graphicx,color}
\usepackage[active]{srcltx}
\usepackage{amsmath} 
\usepackage{amssymb} 
\usepackage{amsmath,amssymb,amsfonts,theorem}
\usepackage{enumerate}
\usepackage{tikz}
\usepackage{tabu}
\usepackage{algpseudocode} 
\usepackage{accents}

\usepackage{epstopdf}

{\theorembodyfont{\normalfont} 

\newcommand\ie{\emph{i.e., }}

\def\t{\color{black}\ } 

\newtheorem{assumption}{Assumption}

\newtheorem{lemma}{Lemma}
\newtheorem{problem}{Problem}
\newtheorem{corollary}{Corollary}
\newtheorem{proposition}{Proposition}

\newcommand{\maria}[1]{\textcolor{black}{#1}}

\def\QEDclosed{\mbox{\rule[0pt]{1.3ex}{1.3ex}}} 
\def\QEDopen{{\setlength{\fboxsep}{0pt}\setlength{\fboxrule}{0.2pt}\fbox{\rule[0pt]{0pt}{1.3ex}\rule[0pt]{1.3ex}{0pt}}}}
\def\QED{\QEDopen}

\title{\LARGE \bf Utility of the Koopman operator in output regulation of disturbed nonlinear systems}

\author{Bart Kieboom, Maria Bartzioka and Matin Jafarian
 \thanks{The authors are with the Delft Center for Systems and Control (DCSC), Delft University of Technology, The Netherlands. Email: {\tt\small b.kieboom@student.tudelft.nl; m.bartzioka@tudelft.nl; m.jafarian@tudelft.nl.} The work of M. Jafarian is supported by the EU-MCIF project ReWoMeN.}
}
\begin{document}
\maketitle
\thispagestyle{empty}
\pagestyle{empty}
\begin{abstract}
This paper studies the problem of output regulation for a class of nonlinear systems experiencing matched input disturbances. It is assumed that the disturbance signal is generated by an external autonomous dynamical system. First, we show that for a class of nonlinear systems admitting a finite-dimensional Koopman representation, the problem is equivalent to a bilinear output regulation. We then prove that a linear dynamic output feedback controller, inspired by the linear output regulation framework, locally solves the original nonlinear problem. Numerical results validate our analysis.  
\end{abstract}
\section{Introduction}
Output regulation is a well-known problem in control theory. The goal is that the system's output asymptotically tracks a reference signal and/or rejects a disturbance signal generated by an external autonomous dynamical system, namely the {\em exosystem}. The problem is well-studied for linear systems \cite{huang2004}, where necessary and sufficient conditions for its solvability depend on the solvability of {\em regulatory equations} \cite{francis1977}. The {\em internal model principle} has been proven essential in solving these equations and designing controllers \cite{francis1975}. Output regulation of nonlinear systems, however, is naturally more challenging. This problem has also been widely studied \cite{isidori1997}. It has been shown that a set of partial differential equations form the regulatory equations for the nonlinear problem \cite{huang2004,isidori1997}. However, finding solutions for the PDEs and designing appropriate internal-model-based controllers for general nonlinear systems and exosystems is quite complex. The problem has been solved for several specific classes of nonlinear systems, and particular control problems \cite{huang2004,serrani2000,byrnes2003}. Owing to its importance and applicability, the nonlinear output regulation problem is still an active field of research. 
In this paper, we propose utilizing the Koopman operator in solving a sub-class of nonlinear output regulation problems. 

The Koopman operator provides an alternative way to model (nonlinear) dynamical systems. It was originally introduced by Bernard Koopman in \cite{koopman1931} and popularized for the study of dynamical systems in \cite{mezic2005,alex2020}. Numerous applications in a wide range of research fields have explored the employment of these techniques \cite{korda2018,brunton2016}. The Koopman operator, associated with a state-space description of a nonlinear dynamical system, is a linear and infinite-dimensional operator that acts on functions of the state of the system, often called observables or observable functions. The action of the Koopman operator on such observables allows us to compute the time evolution of the observables linearly, according to the flow of the system. 

We study matched input disturbance rejection of nonlinear systems, assuming that disturbance signals are generated by linear exosystems. For a class of nonlinear systems that admit a finite-dimensional Koopman representation, we show that the nonlinear problem can be equivalently represented as a bilinear one. Then, we design a linear output feedback controller inspired by the linear output regulation framework to achieve regional bilinear output regulation. We characterize the latter as achieving regional stabilization for the bilinear undisturbed system together with regulating the output error to zero. 
Since the bilinear system equivalently represents the nonlinear system, we conclude that regional bilinear output regulation guarantees that the original nonlinear output regulation problem is locally solved. 



To the best of our knowledge, the existing literature has been mainly focused on bilinear stabilization, both model-based \cite{khlebnikov2016} and data-driven \cite{amato2009,bisoffi2020}. The bilinear output regulation problem, though, has been less explored. Constant disturbance rejection of bilinear systems has been tackled in \cite{isidori1979}. Exploring the applications of the Koopman operator, it has been employed to derive bilinear formulation for nonlinear stabilization problems, in both model-based and data-driven approaches, mainly in the discrete-time setting \cite{korda2018,goswami2017,huang2022}, where the problems have been formulated in an optimal control setting. Compared to the aforementioned works, we solve the continuous-time bilinear matched input disturbance rejection problem, with a general linear exosystem, within the output regulation framework and design a linear controller inspired by the internal model principle. Our approach enables a systematic way for local output regulation of the nonlinear systems.

This paper is organized as follows. Section \ref{section:preliminaries} presents preliminaries on the output regulation problem, Koopman operator theory and the problem formulation. Section \ref{section:norp Koopman} covers the main results of the paper. Section \ref{section:simulation} presents a numerical example. The paper is concluded in section \ref{section:conclusion}.
\section{Preliminaries and Problem Formulation}
\label{section:preliminaries}
In this section, we first recall some required techniques from the output regulation \cite{huang2004} and the Koopman framework \cite{alex2020,mezic2005,brunton2016}, and continue by formulating the problem.
\subsection{Linear Output Regulation}
Consider the following disturbed linear system,
\begin{subequations}
\label{eq:lorp:LORP}
\begin{align}
    \dot{x} &= Ax + Bu + Pv, \label{eq:lorp:LORP a}\\
    y       &= Cx,           \label{eq:lorp:LORP c}
\end{align}
\end{subequations}
with $x\in X\subseteq \mathbb{R}^n$ the state, $u\in\mathbb{R}^m$ the input, $y\in\mathbb{R}^l$ the output error, and $v\in \mathbb{R}^q$ the output of the following linear autonomous system, 
\begin{equation}
\label{eq:pf:exosys plus output}
    \dot{w} = Sw, \quad
    v = Ew,
\end{equation}
with $w\in W\subseteq \mathbb{R}^r$ the exogenous disturbance signal.


The goal is to design a controller such that the undisturbed system has an asymptotically stable equilibrium, and the output error of the system in the presence of the disturbance satisfies 
    $\lim_{t\to\infty} y(t) = 0,$
for any initial condition of the plant and the exosystem.
The exosystem is assumed neutrally stable, \ie $S^T = -S$. 
The general form of a linear dynamic output feedback controller with input $y$ and output $u$ is
\begin{subequations}
\label{eq:lorp:dfc}
\begin{align}
    \dot{\xi} &= F\xi + Gy, \label{eq:lorp:dfc a}\\
    u         &= H\xi + \Gamma y, \label{eq:lorp:dfc b}
\end{align}
\end{subequations}
where $\xi\in \Xi \subseteq \mathbb{R}^\nu$ represents the state of the controller. The closed loop dynamics of the disturbed linear system \eqref{eq:lorp:LORP} influenced by the exosystem \eqref{eq:pf:exosys plus output} and the controller \eqref{eq:lorp:dfc} obeys 
\begin{equation}
\label{eq:lorp:closed loop matrix}
   \begin{bmatrix} \dot x \\ \dot \xi \end{bmatrix}= A_c \begin{bmatrix} x \\ \xi \end{bmatrix} + \begin{bmatrix} P v \\ {\mathbf 0} \end{bmatrix}, \quad  A_c = \begin{bmatrix} (A+B\Gamma C) & BH \\ GC & F \end{bmatrix}.
\end{equation}

It is known that the stabilizability of the pair $(A,B)$ together with the detectability of the pair $(C,A)$ guarantee the existence of matrices  $\{F,G,H,\Gamma\}$ such that $A_c$ is Hurwitz \cite{isidori2003}. Since $A_c$ is Hurwitz, the linear output regulation problem is solved by the controller \eqref{eq:lorp:dfc} if and only if there exist $\Pi\in\mathbb{R}^{n \times r}$ and $\Sigma \in\mathbb{R}^{\nu \times r}$ such that 
\begin{subequations}
\label{eq:lorp:lre and imp}
\begin{align}
    \Pi S    &= (A+B\Gamma C) \Pi + BR + PE, \label{eq:lorp:lre a} \\
    0        &= C\Pi    ,       \label{eq:lorp:lre b} \\
    \Sigma S &= F\Sigma,        \label{eq:lorp:imp a}\\
    R        &= H\Sigma ,       \label{eq:lorp:imp b}
\end{align}
\end{subequations}
are satisfied. The first two equations are called the linear regulator equations \cite{francis1977} and the last two address the internal model principle \cite{francis1975}.
\subsection{Koopman Operator Theory}
Consider the autonomous continuous-time dynamical system described by
\begin{equation}
\label{eq:koopman:autonomous system}
    \dot{x} = f(x),
\end{equation}
with $x\in X\subseteq \mathbb{R}^n$ the state and $f$ a nonlinear function. Integrating \eqref{eq:koopman:autonomous system} yields trajectories $x(t) = F^t(x_0)$, where 
\begin{equation}
    F^t(x_0) = x_0 + \int_{0}^t f(x(\tau))d\tau
\label{eq:koopman:flow}
\end{equation}
is the flow of the system. 

Next, consider functions of the state $\psi:X\to\mathbb{R}$, which are called observable functions.
Denote the space of all such functions as $\mathcal{F}$. The family of Koopman operators $\mathcal{K}^t:\mathcal{F}\mapsto \mathcal{F}$ is defined by 
\begin{equation}
\label{eq:koopman:koopman operator}
    \mathcal{K}^t\psi(x) 
    = \psi(x) \circ F^t 
    = \psi(F^t(x)).
\end{equation}
Since $\mathcal{K}^{t_2}(\mathcal{K}^{t_1}\psi(x)) = \psi(F^{t_1+t_2}(x))$ we simply write $\mathcal{K}$ and refer to it as the Koopman operator. The Koopman operator is infinite-dimensional and linear.

The infinitesimal generator $\mathcal{L}$, of the Koopman operator \cite{lasota1994} is defined by
\begin{equation}
\label{eq:koopman:lie derivative}
    \mathcal{L}\psi 
    = \lim_{t\to 0}\frac{\mathcal{K}^t \psi - \psi}{t} 
    = \lim_{t \to 0}\frac{\psi(F^t(x))-\psi(x)}{t}.
\end{equation}
From \eqref{eq:koopman:lie derivative} we see that $\mathcal{L}$ corresponds to the time derivative of $\psi$ along the trajectories of \eqref{eq:koopman:autonomous system}, \ie $\dot{\psi} = \mathcal{L}\psi$.

The advantage of the Koopman operator is that it allows a nonlinear system to be described linearly. Although such descriptions are generally infinite-dimensional, in some cases a finite, albeit higher dimensional, description can be approximated. A set $\mathcal{D}\subset \mathcal{F}$ of observable functions which satisfies
\begin{equation}
    \mathcal{D} = \{\psi \in \mathcal{F} \mid \psi \in \mathcal{D} \implies \dot{\psi} \in \text{span}(\mathcal{D})\},
\end{equation}
is called a Koopman invariant subspace \cite{brunton2016}. Let
    $\Psi = [\psi_1,\dots,\psi_M]^T$,
with $\psi_i\in\mathcal{D}$. Write 
\begin{equation}
\label{eq:koopman:z}
    z = \Psi(x),
\end{equation}
then it follows that
\begin{equation}
\label{eq:koopman:lti}  
    \dot{z} = Az,
\end{equation}
where $A_{ij}$ is the $j$-th expansion coefficient of $\psi_i$ in $\mathcal{D}$. Hence, a Koopman invariant subspace containing the state components $\psi_i = x_i$ yields the equivalent finite-dimensional linear description \eqref{eq:koopman:lti} of the autonomous nonlinear system \eqref{eq:koopman:autonomous system}, \cite{brunton2016}.
\subsection{Problem Formulation}
This paper considers the utilization of the Koopman operator framework in dealing with matched input disturbances (e.g. \cite{matin1,matin2,matin3}) of nonlinear systems of the form
\begin{subequations}
\label{eq:pf:sys}
\begin{align}
    \dot{x} &= f(x) + g(x)(u+v)  \label{eq:pf:sys a} \\
    y      &= h(x),             \label{eq:pf:sys b}
\end{align}
\end{subequations}
with $x\in X\subseteq \mathbb{R}^n$, $y \in\mathbb{R}^l$, and $f$, $g$ and $h$ smooth, Lipschitz nonlinear functions. In the rest of the paper, we restrict our attention to the scalar input and disturbance case, \ie $u,v \in\mathbb{R}$, and assume that $v$ is the output of the linear exosystem \eqref{eq:pf:exosys plus output}. In the next section, we first characterize a class of nonlinear systems admitting a Koopman representation using a finite number of observables. Therefore, solving the output regulation problem of the bilinear system equivalently solves the output regulation problem of the aforementioned nonlinear system. Inspired by the linear output regulation framework, we design a linear dynamic output feedback controller of the form \eqref{eq:lorp:dfc} for the bilinear system that achieves disturbance rejection and thereby output regulation of the equivalent nonlinear system.
\section{Nonlinear matched input disturbance rejection}
\label{section:norp Koopman}
This section presents the main results of the paper. First, we present a lemma, based on \cite{surana2016,brunton2016,goswami2017}, essential in specifying the class of nonlinear systems that can be represented using finite number of
observables in a Koopman invariant subspace $\mathcal{D}$ \cite{brunton2016}.
\begin{lemma}\label{lem1}
Consider the dynamical system \eqref{eq:pf:sys} evolving in $X\subseteq\mathbb{R}^n$. Suppose there exists a set $\mathcal{D}$ of finite number of observable functions $\psi_i:X\mapsto\mathbb{R}$, $i \in\{1,\ldots,M\}$, with $M=\text{dim}(\text{span}(\mathcal{D})) > n$, that satisfies the following properties:
\begin{enumerate}
    \item if $u=0, v=0$ and $\psi_i \in\mathcal{D}$ then $\dot{\psi_i} \in\text{span}(\mathcal{D}), \forall i\in \{1,\ldots,M\}$,
    \item if $\psi_i \in\mathcal{D}$ and $\frac{d}{dt}(\frac{\partial \psi_i}{\partial x_j}g_{j}) \neq 0$, then $\frac{\partial \psi_i}{\partial x_j}g_{j} \in\text{span}(\mathcal{D}), \forall i \in \{1,\ldots,M\}, \forall j\in \{1,\ldots,n\}$,
    \item it holds that $h_p\in\text{span}(\mathcal{D}), \forall p\in\{1,\ldots,l\},$
    \item it holds that $x_j\in\text{span}(\mathcal{D}), \forall j\in\{1,\ldots,n\}$.
\end{enumerate}
Then the set $\mathcal{D}$ is a Koopman invariant subspace, and the nonlinear system \eqref{eq:pf:sys} is described equivalently by the following bilinear system
\begin{subequations}
\label{eq:lemma1:koopman bilinear system}
\begin{align}
    \dot{z} &= Az + B(u+v) + Nz(u+v), \label{eq:lemma1:koopman bilinear system a} \\
    y       &= Cz, \label{eq:lemma1:koopman bilinear system b}
\end{align}
\end{subequations}
with $z=\Psi(x)\in Z \subseteq \mathbb{R}^M$, where $\Psi = [\psi_1,\dots,\psi_M]^T$, and $A$, $B$, $N$ and $C$ are matrices whose elements are determined by the following equations 
\begin{equation}
    \dot{\psi}_i(x) 
    = \sum_{k=1}^M a_{ik}\psi_{k} 
    + b_{i}(u+v) + \sum_{k=1}^M n_{ik} \psi_k (u+v),
\label{eq:lemma1:psidot2}
\end{equation}
\begin{equation}\label{eq:C}
    h_p (x) = \sum_{k=1}^M c_{pk}\psi_k(x).
\end{equation}
\end{lemma}
\begin{proof}
Choose a set $\mathcal{D}$ of observable functions $\psi$ that satisfies the above properties. For any $\psi_i\in\mathcal{D}$ the time derivative along the trajectories satisfying \eqref{eq:pf:sys} is given by
\begin{equation}
\label{eq:lemma1:psidot}
    \dot{\psi}_i(x) = \sum_{j=1}^n \frac{\partial\psi_i}{\partial x_j }\bigg(f_j(x) + g_{j}(x) (u + v) \bigg).
\end{equation}
Property 1 implies that we can write 
\begin{equation}
\label{eq:lemma1:property 1 implication 1}
    \sum_{j=1}^n \frac{\partial \psi_i}{\partial x_j}f_j(x) = \sum_{k=1}^M a_{ik}\psi_k(x),
\end{equation}
with $a_{ik} \in \mathbb{R}$. Property 2 implies that we can write 
\begin{equation}
\label{eq:lemma1:property 2 implication 1}
    \frac{\partial \psi_i}{\partial x_j}g_{j}(x) = \nu_{ij} + \sum_{k=1}^M \mu_{ijk} \psi_k(x),
\end{equation}
with $\nu_{ij},\mu_{ijk}\in\mathbb{R}$. Using equation \eqref{eq:lemma1:property 2 implication 1} the second term in equation \eqref{eq:lemma1:psidot} may be written as 
\begin{align}
\label{eq:lemma1:property 2 implication 2}
    \bigg(\sum_{j=1}^n \frac{\partial \psi_i}{\partial x_j}g_{j}(x)\bigg)u
    &=  \sum_{j=1}^n \bigg (\nu_{ij} + \sum_{k=1}^M \mu_{ijk} \psi_k(x)\bigg)u \nonumber \\
    &=  b_{i} u + \sum_{k=1}^M n_{ik} \psi_k(x) u,
\end{align}
with $b_{i} = \sum_{j=1}^n \nu_{ij}$ and $n_{ik} = \sum_{j=1}^n \mu_{ijk}$. Substituting \eqref{eq:lemma1:property 1 implication 1} and \eqref{eq:lemma1:property 2 implication 2} in equation \eqref{eq:lemma1:psidot}, we obtain \eqref{eq:lemma1:psidot2}. Let $z = \Psi(x)$ and define the matrices and $A = (a_{ij})$, $B = (b_{i})$ and $N = (n_{ij})$. Using these definitions and \eqref{eq:lemma1:psidot2} we obtain \eqref{eq:lemma1:koopman bilinear system a}. With the dynamics of $z$ known and the initial condition $z_0 = \Psi(x_0)$, property 4 allows the reconstruction of the state $x(t)$ from the state $z(t)$ at each time $t$. Finally, property 3 gives \eqref{eq:C}.
\end{proof}
Notice that the element $b_i$ of the vector $B$ is non-zero only if the function $\frac{\partial \psi_i}{\partial x_j}g_{j}$ has a constant part, that is, $\nu_{ij}$ is non-zero. Typically, the non-zero terms of $B$ appear at the state projections $\psi=x_i$ where the input is added linearly. A simple example is one where $f$ is such that there exists $\mathcal{D}$ that satisfies properties 1 and 4 of Lemma \ref{lem1}, and $g(x) = B$. In the result below, it is crucial that $f$, $g$, and $\mathcal{D}$ are such that non-zero elements of $B$ appear that can be used for control. Assumption \ref{as1} (see below) captures this requirement. 
\subsection{Bilinear matched input disturbance rejection}
This section utilizes the equivalence of systems \eqref{eq:pf:sys} and \eqref{eq:lemma1:koopman bilinear system}, the former a nonlinear system satisfying Lemma \ref{lem1} and the latter a bilinear system, in tackling the nonlinear output regulation problem. 
\begin{problem}[Regional bilinear output regulation] Consider the bilinear system \eqref{eq:lemma1:koopman bilinear system} subject to a matched input disturbance generated by the linear exosystem \eqref{eq:pf:exosys plus output}. The problem is that given a compact set of initial conditions whether there exists a linear output feedback controller capable of achieving regional stability for the undisturbed system, \ie ensuring that trajectories originating from that compact set converge to an asymptotically stable equilibrium within the set, while also regulating the output of the closed-loop disturbed system, that is $\lim_{t\to\infty} y(t) = 0$.
\end{problem}

\begin{assumption}\label{as1}
For the bilinear system \eqref{eq:lemma1:koopman bilinear system}, the pair $(A,B)$ is stabilizable, and the pair $(A,C)$ is detectable.
\end{assumption}
In the spirit of linear output regulation, we will look for a trajectory $(z(t),\xi(t),w(t)) = (\Pi w(t), \Sigma w(t), w(t))$ of the bilinear system \eqref{eq:lemma1:koopman bilinear system} influenced by the exosystem \eqref{eq:pf:exosys plus output} and the controller \eqref{eq:lorp:dfc}, that satisfies the closed-loop dynamics. Thus, the following equations must be satisfied
\begin{subequations}
\label{eq:borp:bre}
\begin{align}
    \Pi S    &= (A+B\Gamma C)\Pi\!\!+\!(B+N \Pi w)(H \Sigma+E)+ N\Pi w \Gamma C \Pi, \label{eq:borp:bre a} \\
    0        &= C\Pi,           \label{eq:borp:bre b} \\
    \Sigma S &= F\Sigma + GC\Pi,        \label{eq:borp:bre c}\\
    R        &= H\Sigma. \label{eq:borp:bre d}
\end{align}
\end{subequations}
Writing the above equations in the compact form, we obtain
\begin{equation}
\label{eq:sylv eq}
    \begin{bmatrix} \Pi \\ \Sigma \end{bmatrix} S -
    \underbrace{\begin{bmatrix} A+B\Gamma C & BH \\ GC & F \end{bmatrix}}_{A_c}
    \begin{bmatrix} \Pi \\ \Sigma \end{bmatrix} =
    \begin{bmatrix} BE \\ 0 \end{bmatrix} + \gamma(w),
\end{equation}
where $\gamma(w)= N \Pi w (H \Sigma+E)+ N \Pi w \Gamma C \Pi$. Notice that Assumption \ref{as1} guarantees that the closed-loop matrix \eqref{eq:lorp:closed loop matrix} can be made Hurwitz \cite{isidori2003}. We design $H\Sigma + E = 0$. Since $C\Pi=0$, it holds that $\gamma(w)=0$, and the remaining matrix equality in \eqref{eq:sylv eq} is the Sylvester equation. Since the spectrum of $S$ lies on the imaginary axis and $A_c$ is Hurwitz, $\Pi = 0$ is the unique solution satisfying \eqref{eq:sylv eq}. 
Thus, we consider the error dynamics of $(\dot z, \dot{\tilde{\xi}})$, where $\tilde{\xi} = \xi - \Sigma w$. Consider the bilinear system in \eqref{eq:lemma1:koopman bilinear system} with the linear exosystem \eqref{eq:lorp:LORP}. The closed-loop error dynamics obey 
\begin{subequations}
\label{eq:prop2:cl1}
\begin{align}
    \begin{split}
    \dot{z} &= (A+B\Gamma C)z + B H\tilde{\xi} + Nz H\tilde{\xi}+Nz\Gamma Cz \\
            &\hspace{12pt} + B(H\Sigma + E)w + Nz(H\Sigma + E)w,
    \end{split} \\
    \dot{\tilde{\xi}} &= F\tilde{\xi} + GCz + (F\Sigma + \Sigma S)w, \\
    \dot{w} &= Sw.
\end{align}
\end{subequations}

Let $p(t)$ denotes the closed-loop state $\begin{bmatrix} z(t) \\ \xi(t) - \Sigma w(t) \end{bmatrix}$.\\ 
\begin{proposition}\label{pr1}
Consider the bilinear system in \eqref{eq:lemma1:koopman bilinear system} under Assumption \ref{as1}, and subject to the matched input disturbance generated by the exosystem \eqref{eq:pf:exosys plus output}. The linear dynamic output feedback controller \eqref{eq:lorp:dfc} solves the regional bilinear output regulation problem for the initial conditions $p(0)\in \mathcal{I}(w)=\{p \in Z \times \Xi : p^T W^{-1} p \leq 1\}$, with $0 \prec W=W^T$, provided that the following conditions hold: 
\begin{enumerate}
\item the internal model principle is satisfied, \ie there exist a matrix $\Sigma$, such that $\Sigma S=F \Sigma$,
\item $H\Sigma + E = 0$ holds,
\item 
$\exists \epsilon >0$ such that 
$W \succ 0$ satisfies
\begin{equation}
    \begin{bmatrix}
    WA_c^T + A_cW + \epsilon \tilde{N}W\tilde{N}^T &     
    W \tilde{H}^T\\
    \tilde{H}W & -\epsilon I 
    \end{bmatrix} \prec 0,
\label{eq:prop2:matrix ineq}
\end{equation}
where     
    $\tilde{N} = \begin{bmatrix} N & 0 \\ 0 & 0\t \end{bmatrix}\ \text{and}\  
\tilde{H} = \begin{bmatrix} \Gamma C & H \end{bmatrix}.$
\end{enumerate}
\end{proposition}

\begin{proof}
Consider the closed-loop dynamics \eqref{eq:lemma1:koopman bilinear system} with $v$ generated by \eqref{eq:pf:exosys plus output}. Assuming that the internal model principle is satisfied, \ie  $\Sigma S=F \Sigma$, and $H\Sigma +E = 0$ holds, the closed-loop error dynamics \eqref{eq:prop2:cl1} reduces to
\begin{subequations}
\label{eq:prop2:cl2}
\begin{align}
    \dot{z}   &= (A+B\Gamma C)z + B H\tilde{\xi} + Nz H\tilde{\xi}+ Nz\Gamma Cz, \label{eq:prop2:cl2 b} \\
    \dot{\tilde{\xi}} &= F\tilde{\xi} + GCz. \label{eq:prop2:cl2 a} 
\end{align}
\end{subequations}
Recall that the closed-loop matrix, $A_c$, is defined in \eqref{eq:lorp:closed loop matrix}. Considering $A_c$ and using the definitions $\tilde{N} = \begin{bmatrix} N & 0 \\ 0 & 0\t \end{bmatrix}\ \text{and}\  
\tilde{H} = \begin{bmatrix} \Gamma C & H \end{bmatrix}$, the closed-loop error dynamics \eqref{eq:prop2:cl2} can be re-written in the following compact form
\begin{equation}
\label{eq:prop2:cl3}
    \dot{p} = A_c p + \tilde{N} p \tilde{H} p,
\end{equation}
where $p(t)$ is the state of the closed loop system. Since the matrix $A_c$ is Hurwitz, considering  Assumption \ref{as1} and the controller design, the linear part of the dynamics \eqref{eq:prop2:cl3} is asymptotically stable. Now, define the following quadratic Lyapunov function \cite{khlebnikov2016}
\begin{equation}
    V(p) = p^T W^{-1} p.
\label{eq:prop2:quadratic lyapunov function}
\end{equation}
The time derivative of \eqref{eq:prop2:quadratic lyapunov function} along the trajectories of the closed-loop dynamics \eqref{eq:prop2:cl3} gives
\begin{equation}
\begin{aligned}
    \dot{V}(p) =&\, p^T\big(A_c^TW^{-1} + W^{-1}A_c \\
    &+ \tilde{H}^T p^T\tilde{N}^TW^{-1} + W^{-1}\tilde{N} p \tilde{H}\big) p.
\end{aligned}
\end{equation}
In what follows, using a modification of Petersen’s lemma \cite{khlebnikov2016}, we characterize conditions under which $\dot{V}(p) <0$ holds, that is
\begin{equation}
\label{eq:prop2:inequality 1}
\begin{aligned}
    &A_c^TW^{-1} + W^{-1}A_c \\
     & + \tilde{H}^Tp^T\tilde{N}^TW^{-1} + W^{-1}\tilde{N}p\tilde{H} \prec 0.
\end{aligned}
\end{equation}
Pre- and post multiplying the inequality \eqref{eq:prop2:inequality 1} with $W$ yields the equivalent inequality
\begin{equation}
\label{eq:prop2:inequality 2}
\begin{aligned}
    WA_c^T + A_c W + W\tilde{H}^T p ^T\tilde{N}^T + \tilde{N} p\tilde{H}W \prec 0.
\end{aligned}
\end{equation}
Let $T = WA_c^T + A_cW$  and $Y=\tilde{H}W$, then the inequality \eqref{eq:prop2:inequality 2} is written as
\begin{equation}
    T + Y^T p ^T \tilde{N}^T + \tilde{N}p Y \prec 0,
\label{eq:prop2:intermediate inequality}
\end{equation}
with $T = T^T$. Since $p(0)\in \mathcal{I}(w)$ \maria{and $0 \prec W=W^T$} it follows by Lemma $2$ in \cite{khlebnikov2016} that the inequality \eqref{eq:prop2:inequality 2} is satisfied if and only if there exist $\epsilon > 0$ such that the following inequality holds
\begin{equation}
    \begin{bmatrix}
    T + \epsilon \tilde{N}W\tilde{N}^T & Y^T \\
    Y & -\epsilon I \\
    \end{bmatrix} \prec 0.
\label{eq:prop2:LMI1}
\end{equation}
The above inequality is exactly the condition in \eqref{eq:prop2:matrix ineq}. As a result, any trajectory $p(t)$ for which $p(0) \in\mathcal{I}(w)$ will converge to the origin, which means $z(t)$ converges to zero and
\begin{equation}
    \lim_{t\to\infty} y(t) = \lim_{t\to\infty} Cz(t) = 0.
\end{equation}
Hence, the disturbance is rejected and the output is regulated. 
\end{proof}
\subsection {Nonlinear matched input disturbance rejection}
The equivalence of the nonlinear dynamical system \eqref{eq:pf:sys} and the bilinear dynamical system \eqref{eq:lemma1:koopman bilinear system} leads to a systematic approach guaranteeing that the problem of nonlinear matched input disturbance rejection is locally solved. 

\begin{corollary}\label{cor1}
Consider the nonlinear system \eqref{eq:pf:sys} subject to matched input disturbance generated by the exosystem \eqref{eq:pf:exosys plus output} admitting \maria{an equivalent} bilinear representation by the Koopman operator, based on Lemma \ref{lem1}. Then, the nonlinear output regulation problem is locally solved by the linear dynamic output feedback controller \eqref{eq:lorp:dfc} provided that conditions of Proposition \ref{pr1} are satisfied for its equivalent bilinear system.
\end{corollary}

\subsection*{Controller synthesis}
Proposition \ref{pr1} gives conditions guaranteeing that a general form of a linear controller successfully deals with the regional bilinear matched input disturbance problem. Depending on the available information, several designs and tuning could be proposed for the linear controller. For this, we refer to the wide literature on designing linear output regulators \cite{isidori1997, isidori2003,huang2004}.  

In the next section, for a given nonlinear system, we design a linear controller, satisfying the conditions of Proposition \ref{pr1}, with output $u(t)=u_1(t)+u_2(t)$, where $u_1(t)$ rejects the disturbance, and $u_2(t)$ stabilizes the closed-loop matrix $A_c$.
\section{Simulation results}
\label{section:simulation}
In this section, we numerically validate the results of Proposition \ref{pr1} and Corollary \ref{cor1}. Inspired by \cite{kaiser2021}, we consider the nonlinear system
\begin{subequations}
\label{eq:sim:example sys}
\begin{align}
    \dot{x}_1 &= \kappa_1 x_1 + x_1(u + v), \label{eq:sim:example sys a}\\
    \dot{x}_2 &= \kappa_2(x_2 - x_1^2) + u + v, \label{eq:sim:example sys b} \\
    y         &= x_2 - \frac{1}{6}x_2^3, \label{eq:sim:example sys e}
\end{align}
\end{subequations}
where $v$ is a scalar and the output of the sinusoidal exosystem
\begin{equation}
    \dot{w}_1 = -\kappa_3 w_2, \  
    \dot{w}_2 = \kappa_3 w_1, \ 
    v  = w_1.
\end{equation}
We set the plant and exosystem parameters as $\kappa_1=-0.7$, $\kappa_2 = -0.3$ and $\kappa_3 = 4$.  As the nonlinear model \eqref{eq:sim:example sys} satisfies all the assumptions of Lemma \ref{lem1}, we built a dictionary in order to obtain an equivalent bilinear dynamical model. The respective matrices, $A,\;B,\;N$ and $C$ are provided in the Appendix. Using the Popov- Belevitch- Hautus (PBH) test, we confirm that Assumption \ref{as1} is satisfied for the bilinear system. Following Proposition \ref{pr1}, we design a linear dynamic output feedback controller \eqref{eq:lorp:dfc} as follows: 
   $$ F =S=\begin{bmatrix}
    0 & -4 \\
    4 & 0
 \end{bmatrix},\quad
   H =-E=\begin{bmatrix}
    -1 & 0
\end{bmatrix},\quad
    G =-E^T,$$
and choose a scalar $\Gamma=500$ such that 
$A_c$ is Hurwitz. We fix the number $\epsilon = 0.01$ and determine positive definite $W$ such that the matrix inequality \eqref{eq:prop2:matrix ineq} is satisfied. 
We set the initial conditions of the nonlinear system, exosystem, and the controller as $x(0)=[1\;1]^T; w(0)=[1\;1]^T; \xi(0)=[1\;1]^T$. Then, the nonlinear system will yield an initial output value $y_n(0)=5/6$. Using the Koopman mapping, the initial condition for the equivalent bilinear system is $z(0)=\begin{bmatrix} 1& 1 & 1 & 1 & 1 & 1 & 1 & 1 & 1 & 1\end{bmatrix}^T$. We verify that given $x(0)$ and $\xi(0)$, $z(0)$ satisfies \eqref{eq:koopman:z} and $p(0)$ is within the basin of attraction of the bilinear system $\mathcal{I}(w)$. 

Figure \ref{fig:undis} shows the time evolution of the bilinear system's output, 
$y_b(t)$, as well as the output of the controller in the absence of disturbances. As shown the output and control law converge to zero. Figure \ref{fig:disb} shows time evolution of the sinusoidal disturbance, $v(t)$, which affects both the nonlinear system and its equivalent bilinear system, together with the bilinear system output and the controller output, $u(t)$. As shown despite the effect of the disturbance, the system output is regulated to zero. Figure \ref{fig:disn} shows the output of the controlled disturbed nonlinear system, 
$y_n(t)$, together with the system's states. As shown the designed controller stabilizes the nonlinear system, and its output is regulated to zero.

\begin{figure}[h!]
    \centering
    \includegraphics[scale=0.51]{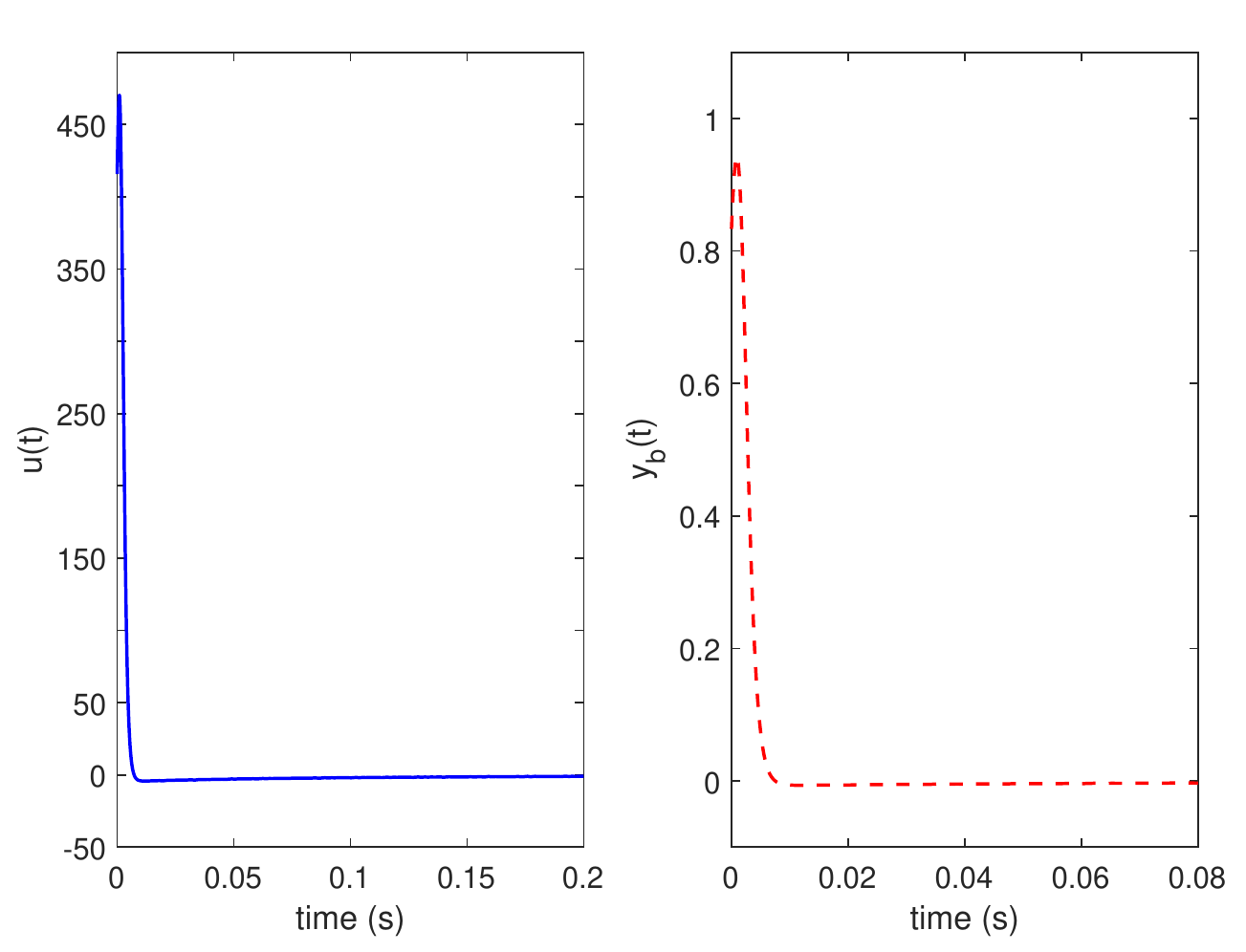}
    \caption{Time evolution of the control law and system output of the undisturbed bilinear system.}
    \label{fig:undis}
\end{figure}

\begin{figure}[h!]
    \centering
    \includegraphics[scale=0.51]{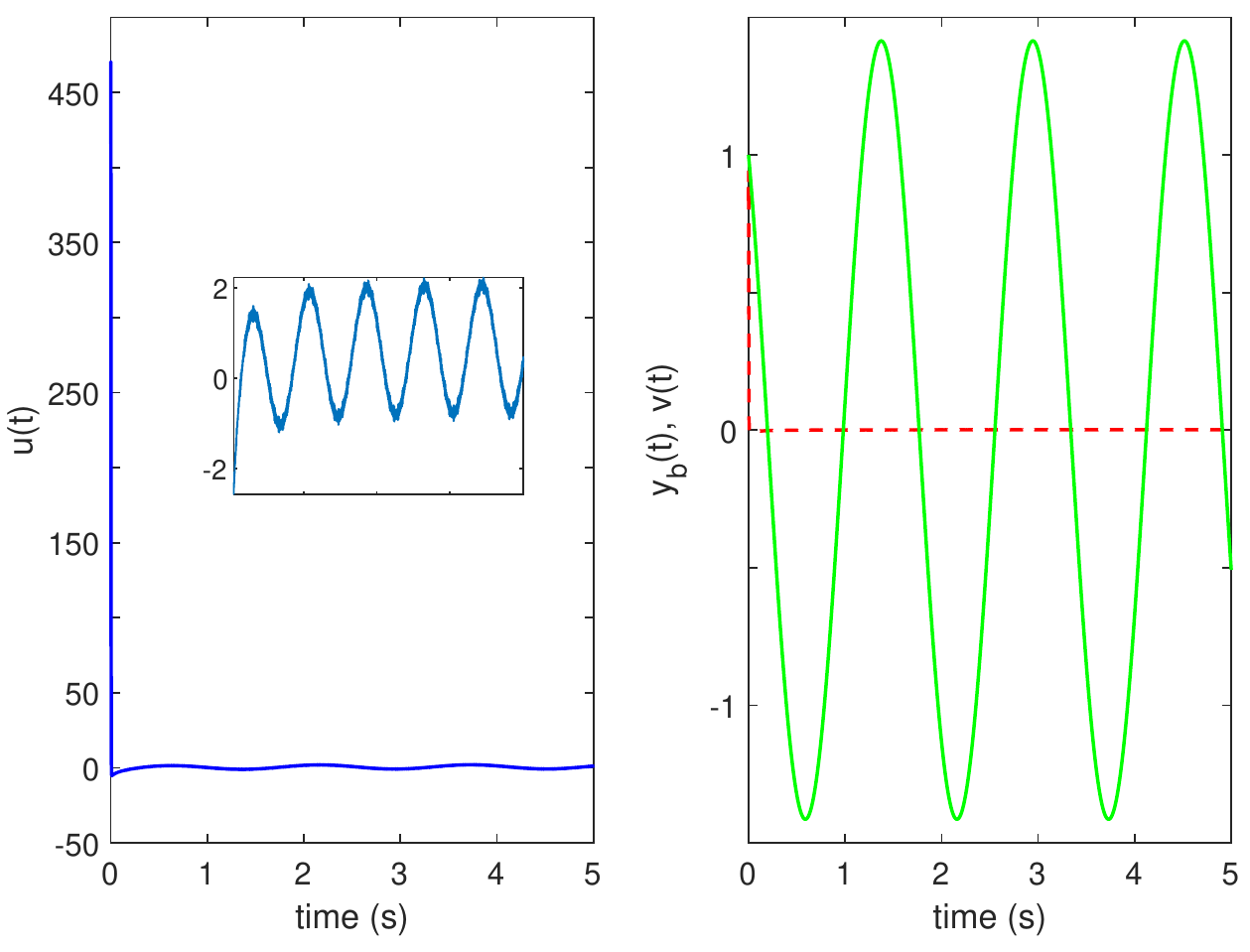}
    \caption{Time evolution of the disturbance signal, control law and output of the disturbed bilinear system.}
    \label{fig:disb}
\end{figure}

\begin{figure}[h!]
    \centering
    \includegraphics[scale=0.51]{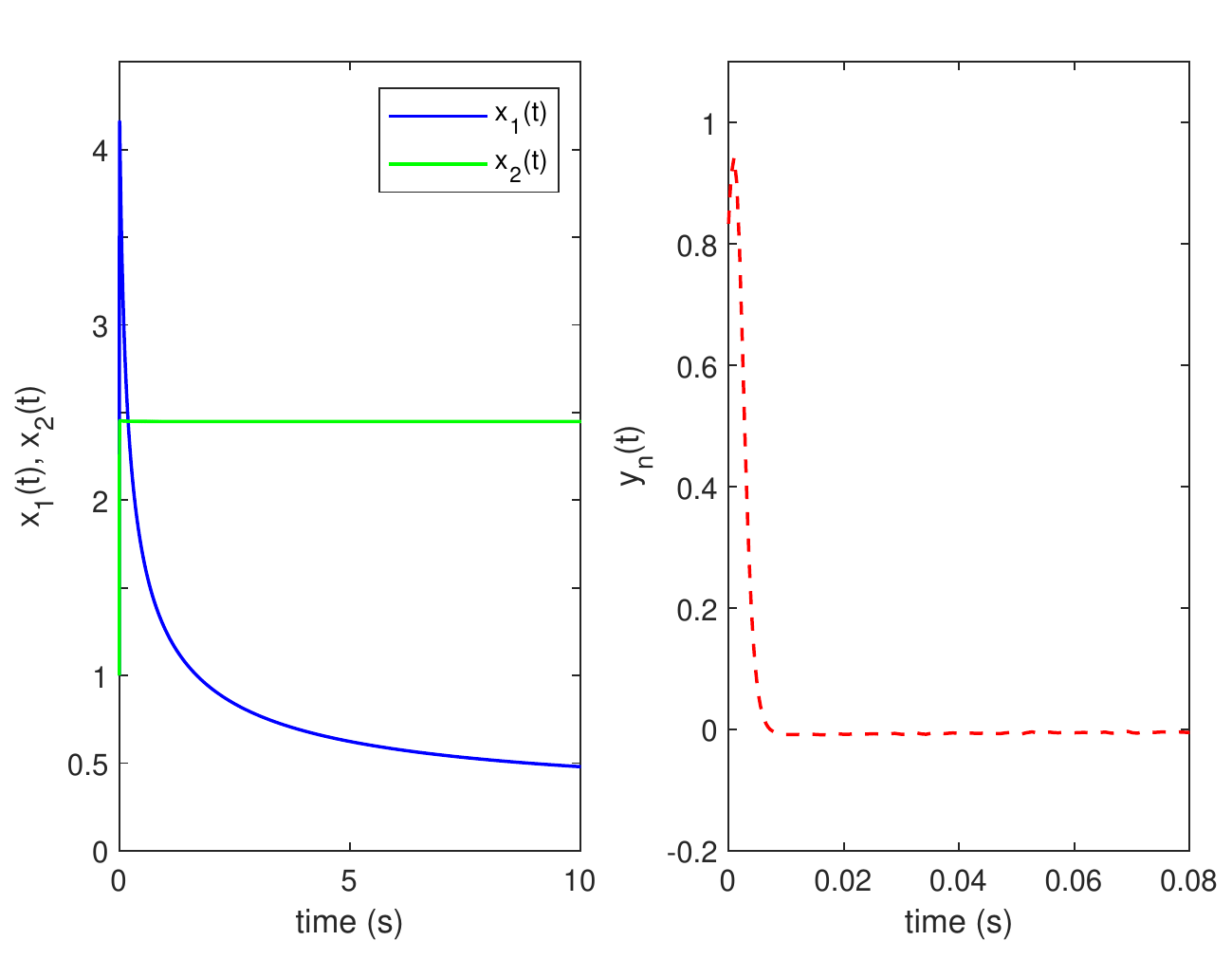}
    \caption{Time evolution of the states and output of the disturbed nonlinear system.}
    \label{fig:disn}
\end{figure}
\section{Conclusion}
\label{section:conclusion}
In this paper, we proposed a novel approach for matched input disturbance rejection of nonlinear systems. A class of nonlinear control systems that admit
a finite-dimensional Koopman representation was considered. We proved that a linear dynamic output feedback controller can solve the regional bilinear matched input disturbance rejection problem, hence, solving the original nonlinear problem locally. Future research directions include extending this framework to tackle reference tracking and address projection errors arising from approximations of the bilinear model. Additionally, exploring data-driven methods using Koopman operator for output regulation is another interesting direction.

\section*{APPENDIX}
To find a Koopman bilinear model for the system \eqref{eq:sim:example sys}, we use the following dictionary $\mathcal{D}=\{\psi_1,\dots,\psi_{10}\}$, with
\begin{align*}
    &\psi_1(x) = x_1,\         
    \psi_2(x) = x_2,\ 
    \psi_3(x) = x_1^2, \\    
    &\psi_4(x) = x_2^2, \ 
    \psi_5(x) = x_2^3, \ 
    \psi_6(x) = x_1^2x_2, \\
    &\psi_7(x) = x_1^4,\  
    \psi_8(x) = x_1^2x_2^2,\ 
    \psi_9(x) = x_1^4x_2,\   
    \psi_{10}(x) = x_1^6.
\end{align*}
One can verify that such a dictionary satisfies the properties of Lemma \ref{lem1} for the nonlinear system given by \eqref{eq:sim:example sys}.

The system matrices of the bilinear Koopman model \eqref{eq:lemma1:koopman bilinear system} associated with the dictionary are obtained by computing the time derivative along the dynamics \eqref{eq:sim:example sys} for each $\psi_i \in \mathcal{D}$. Thus, if $\psi = x_1^n x_2^m$ and $m,n\geq1$, then
\begin{align}
\label{eq:app:time derivative expression}
    \dot{\psi}
    &= n x_1^{n-1} x_2^m \dot{x}_1 + m x_1^n x_2^{m-1} \dot{x}_2 \nonumber\\
    \begin{split}
    &= (n \kappa_1 + m \kappa_2) x_1^n x_2^m - m \kappa_2 x_1^{n+2}x_2^{m-1} \\
    &\hspace{11pt} + (nx_1^nx_2^m + mx_1^nx_2^{m-1})(u+v)
    \end{split}
\end{align}
Moreover the bilinear system matrices are provided below.
\begin{equation*}
    N = 
    \begin{bmatrix}
    1 &  0 & 0 &  0 & 0 &  0 & 0 &  0 & 0 & 0 \\
    0 &  0 & 0 &  0 & 0 &  0 & 0 &  0 & 0 & 0 \\
    0 &  0 & 2 &  0 & 0 &  0 & 0 &  0 & 0 & 0 \\
    0 & 2 & 0 &  0 & 0 &  0 & 0 &  0 & 0 & 0 \\
    0 &  0 & 0 & 3 & 0 &  0 & 0 &  0 & 0 & 0 \\
    0 &  0 & 1 &  0 & 0 &  2 & 0 &  0 & 0 & 0 \\
    0 &  0 & 0 &  0 & 0 &  0 & 4 &  0 & 0 & 0 \\
    0 &  0 & 0 &  0 & 0 & 2 & 0 & 2 & 0 & 0 \\
    0 &  0 & 0 &  0 & 0 &  0 & 1 &  0 & 4 & 0 \\
    0 &  0 & 0 &  0 & 0 &  0 & 0 &  0 & 0 & 6 
    \end{bmatrix}.
\end{equation*}
Matrix $A\in\mathbb{R}^{10\times 10}$ is composed of the following rows: 
\begin{align*}
	A_1= & [\kappa_1\quad 0\quad 0\quad 0\quad 0\quad 0\quad 0\quad 0\quad 0\quad 0],\\
	A_2= & [0\quad \kappa_2\quad -\kappa_2\quad 0\quad 0\quad 0\quad 0\quad 0\quad 0\quad 0],\\
	A_3= & [0\quad 0\quad 2\kappa_1\quad 0\quad 0\quad 0\quad 0\quad 0\quad 0\quad 0],\\
    A_4= & [0\quad 0\quad 0\quad 2\kappa_2\quad 0\quad -2\kappa_2\quad 0\quad 0\quad 0\quad 0],\\
    A_5= & [0\quad 0\quad 0\quad 0\quad 3\kappa_2\quad 0\quad 0\quad -3\kappa_2\quad 0\quad 0],\\
    A_6= & [0\quad 0\quad 0\quad 0\quad 0\quad 2\kappa_1+\kappa_2\quad -\kappa_2\quad 0\quad 0\quad 0],\\
    A_7= & [0\quad 0\quad 0\quad 0\quad 0\quad 0\quad 4\kappa_1\quad 0\quad 0\quad 0],\\
    A_8= & [0\quad 0\quad 0\quad 0\quad 0\quad 0\quad 0\quad  2\kappa_1+2\kappa_2 \quad -2\kappa_2 \quad 0],\\
	A_9= & [0\quad 0\quad 0\quad 0\quad 0\quad 0\quad 0\quad 0\quad 4\kappa_1+\kappa_2\quad -\kappa_2],\\
	A_{10}= & [0\quad 0\quad 0\quad 0\quad 0\quad 0\quad 0\quad 0\quad 0\quad 6\kappa_1],
\end{align*}
Finally, matrices $B$ and $C$ are
\begin{equation*}
    B = 
    \begin{bmatrix}
    0 & 1 & 0 & 0 & 0 & 0 & 0 & 0 & 0 & 0
    \end{bmatrix}^T,
\end{equation*}
\begin{equation*}
    C = 
    \begin{bmatrix}
    0 & 1 & 0 & 0 & -\frac{1}{6} & 0 & 0 & 0 & 0 & 0
    \end{bmatrix}.
\end{equation*}
\bibliographystyle{ieeetr}
\bibliography{biblio}

\end{document}